\documentclass{article}
\usepackage{ijcai24}

\usepackage{times}
\usepackage{soul}
\usepackage{url}
\usepackage[hidelinks]{hyperref}
\usepackage[utf8]{inputenc}
\usepackage[small]{caption}
\usepackage{graphicx}
\usepackage{amsmath}
\usepackage{amsthm}
\usepackage{booktabs}
\usepackage[linesnumbered,ruled,vlined]{algorithm2e}
\usepackage[switch]{lineno}

\usepackage[capitalize]{cleveref}
\allowdisplaybreaks[3]

\newcommand{\citet}[1]{\citeauthor{#1}~\shortcite{#1}}
\newcommand{\citep}{\cite}

\newcommand{\propa}{PROP$\alpha$}
\newcommand{\propfa}{PROP$f(\alpha)$}
\newcommand{\efa}{EF$\alpha$}
\newcommand{\effa}{EF$f(\alpha)$}

\theoremstyle{definition}
\newtheorem{theorem}{Theorem}[section]
\newtheorem{definition}[theorem]{Definition}
\newtheorem{claim}[theorem]{Claim}
\newtheorem{lemma}[theorem]{Lemma}
\newtheorem{observation}[theorem]{Observation}

\newtheorem{corollary}[theorem]{Corollary}

\title{Allocating Mixed Goods with Customized Fairness and Indivisibility Ratio}

\author{
Bo Li$^1$\and
Zihao Li$^2$\and
Shengxin Liu$^{3}$\And
Zekai Wu$^3$
\affiliations
$^1$The Hong Kong Polytechnic University\\
$^2$Nanyang Technological University\\
$^3$Harbin Institute of Technology, Shenzhen\\
\emails
comp-bo.li@polyu.edu.hk,
zihao004@e.ntu.edu.sg,
\{sxliu@, 200110607@stu.\}hit.edu.cn
}

\begin{document}

\maketitle

\begin{abstract}
    We consider the problem of fairly allocating a combination of divisible and indivisible goods.
    While fairness criteria like envy-freeness (EF) and proportionality (PROP) can always be achieved for divisible goods, only their relaxed versions, such as the ``up to one'' relaxations EF1 and PROP1, can be satisfied when the goods are indivisible. 
    The ``up to one'' relaxations require the fairness conditions to be satisfied provided that one good can be completely eliminated or added in the comparison. 
    In this work, we bridge the gap between the two extremes and propose ``up to a fraction'' relaxations for the allocation of mixed divisible and indivisible goods.
    The fraction is determined based on the proportion of indivisible goods, which we call the indivisibility ratio.
    The new concepts also introduce asymmetric conditions that are customized for individuals with varying indivisibility ratios. 
    We provide both upper and lower bounds on the fractions of the modified item in order to satisfy the fairness criterion. 
    Our results are tight up to a constant for EF and asymptotically tight for PROP. 
\end{abstract}

\section{Introduction}\label{sec:intro}
Fair division of a mixture of divisible and indivisible goods has been well motivated since \citet{BeiLiLi20}. This scenario is exemplified in the context of dividing inheritances, where the assets include both money and land (divisible goods) as well as houses and cars (indivisible goods).
In contrast to the division of purely divisible goods, one of the key challenges lies in defining and characterizing the notions of fairness that are both ideal and practical.
This aspect continues to be a subject of ongoing discussion and exploration in the literature \citep{KawaseNiSu23,LiuLuSu23}, and 
our work contributes to this ongoing debate.

When the goods are all divisible, {\em envy-freeness} (EF) \citep{Foley67,varian1973equity} and {\em proportionality} (PROP) \citep{Steinhaus49} are the prominent fairness notions. 
Informally, an allocation is EF if every agent does not envy any other agent's bundle, and is PROP if every agent's utility over her bundle is no less than $1/n$, where $n$ is the number of agents throughout the paper and each agent's total utility is normalized to 1.  
When the goods are all indivisible, due to the fact that EF and PROP allocations barely exist, the ``up to one'' relaxation is one of the most widely accepted notions, such as {\em envy-freeness up to one good} (EF1) which ensures that the envy between two agents can be resolved by removing a single good \citep{LiptonMaMo04,Budish11}, and {\em proportionality up to one good} (PROP1) which requires every agent's utility to be no less than $1/n$ after grabbing an additional good from some other agent's bundle \citep{CFS17}. EF1 and PROP1 have some nice properties, such as guaranteed existence, simple computation, and being compatible with Pareto optimality (PO) \citep{CaragiannisKuMo19}. 

When the goods are mixed, EF1 and PROP1 can be directly applied and guaranteed to be satisfiable, by treating the divisible goods as hypothetical infinitesimally indivisible units.
However, these ``up to one'' relaxations are rather weak fairness criteria, as the presence of divisible goods can help alleviate the burden of unfairness. 
In light of this, \citet{BeiLiLi20} introduced {\em envy-freeness for mixed goods} (EFM), whose existence is also guaranteed: for any two agents $i$ and $j$, agent $i$ does not EFM-envy agent $j$ if (1) agent $i$ does not envy agent $j$, or (2) agent $j$'s bundle only contains indivisible goods and agent $i$ does not EF1-envy agent $j$.
EFM serves as a stronger notion than EF1 as condition (2) forces $j$'s bundle to only contain indivisible goods if agents $i$ envies $j$.

Apart from EFM, a more straightforward approach to enhance EF1 and quantify the help of divisible goods in achieving fairness is to directly strengthen the ``up to one'' relaxation to the ``up to a fraction'', and the specific fraction depends on the portion of indivisible goods in relation to all goods. 
Intuitively, an agent may desire fairer allocations when her portion of divisible goods is more valuable.
One possible way to quantify the portion of (in)divisible goods for each agent $i$ is through her \emph{indivisibility ratio} $\alpha_i$, where  $\alpha_i$ represents the portion of utility derived from indivisible goods.
Then, an allocation is \emph{envy-free up to $\alpha$-fraction of one good} (EF$\alpha$) to agent $i$ if any envy she has towards another agent $j$ can be resolved by obtaining an $\alpha_i$ fraction of some indivisible item from agent $j$'s bundle.
Similarly, an allocation is \emph{proportional up to $\alpha$-fraction of one good} (PROP$\alpha$) to agent $i$ if her utility remains at least $1/n$ after acquiring an $\alpha_i$ fraction of some indivisible item from another agent's bundle.
It is important to note that the ``up to $\alpha$'' relaxation allows for varying indivisibility ratios among the agents, thereby tailoring the evaluation of fairness based on each agent's specific perspective. 
In this work, we focus on EF$\alpha$ and PROP$\alpha$.

\paragraph{Example}
To illustrate the difference between EFM and EF$\alpha$, we consider the following example, where two indivisible goods $M=\{o_1, o_2\}$ and one cake $C$ are allocated to three identical agents.
The utility function $u(\cdot)$ is shown in Table \ref{table:example:intro}.
Allocation $\mathcal{A} = (A_1, A_2,A_3)$ with $A_1=A_2=\frac12C$ and $A_3=\{o_1,o_2\}$ is EFM, but it is not EF$\alpha$ since removing $0.5$ fraction from item $o_1$, the remaining utility of $A_3$ is $0.25\times0.5+0.25=0.375$ which is still greater than $0.25$.

From this example, we observe that when the indivisibility ratio is small, EF$\alpha$ can ensure a fairer or more balanced allocation which can be closer to EF. In contrast, EFM may return an allocation that appears somewhat unfair due to its adherence to the EF1 criteria for bundles comprising solely indivisible goods. Furthermore, unlike EFM, when agents are non-identical, EF$\alpha$ guarantees customized fairness based on various personalized indivisibility ratios.

\begin{table}[ht]
    \centering
    \begin{tabular}{c|cc|c|c}
        \toprule
        &  $o_1$ & $o_2$ & $C$ & $\alpha$ \\
        \midrule
        $u(\cdot)$ &  $0.25$ & $0.25$ & $0.5$ & $0.5$\\
        \bottomrule
    \end{tabular}
    \caption{An Example on EFM v.s. EF$\alpha$}
    \label{table:example:intro}
\end{table}

\subsection{Main Results}
In this paper, we propose to study the ``up to a fraction'' relaxation of EF and PROP, when a mixture of divisible and indivisible goods are allocated. We show that an EF$\alpha$ allocation may not exist and a PROP$\alpha$ allocation always exists. Thus we would like to understand to what extent EF$\alpha$ needs to be relaxed and PROP$\alpha$ can be strengthened, namely EF$f(\alpha)$ and PROP$f(\alpha)$, so that a fair ``up to a fraction'' allocation exists.

In Section~\ref{sec:efalpha}, we study the ``up to a fraction'' relaxation of EF, i.e., EF$\alpha$ and EF$f(\alpha)$.
We first prove that $f(\alpha)=\Theta(n)\alpha$ is necessary and sufficient to satisfy EF by removing $f(\alpha)$ fraction of a good. 
We find that any EFM allocation is EF$n\alpha$, and thus an EF$n\alpha$ allocation always exists (by \citep{BeiLiLi20}). 
The guarantee of EFM cannot be improved even when agents have identical valuations.
On the other hand, we prove that at least $\frac{n^2}{4(n-1)}\alpha$ fraction of the good has to be removed in order to satisfy EF, and thus our results are tight up to a constant.
Besides, when agents have identical valuations, we show that a simple greedy algorithm ensures an EF$\frac{n^2}{4(n-1)}\alpha$ allocation, which exactly characterizes the extent to which EF$f(\alpha)$ can be guaranteed in this restricted case.

We then focus on the ``up to a fraction'' relaxation of PROP, namely, PROP$\alpha$ and PROP$f(\alpha)$, in Section~\ref{sec:propalpha}.
In contrast to EF$\alpha$, EFM implies PROP$\alpha$ whose existence is thus guaranteed. 
Additionally, we design a simple polynomial-time algorithm to compute such an allocation.
On the negative side, we find that a PROP$(\frac{n-1}{n}-\varepsilon)\alpha$ allocation does not always exist for any $\varepsilon > 0$ so that our bound is asymptotically the best possible. 
On top of the existence result, in Section~\ref{sec:propa+po}, we prove that PROP$\alpha$ and an economic efficiency criterion of Pareto optimality (PO) can be ensured simultaneously. 
Throughout our analysis, we draw upon the ideas in \citep{CaragiannisKuMo19} and show that any maximum Nash welfare allocation satisfies PROP$\alpha$ and PO in the context of mixed goods.
We stress that our analysis significantly differs from and is much more intricate than the analysis in \citep{CaragiannisKuMo19}. 
This is because we need to evaluate the allocation as a whole, compelled by the PROP requirement and the definition of the indivisibility ratio, rather than relying on a simple pairwise exchange analysis.
To tackle this issue, we utilize a monotone property to derive a decent tight condition on one agent's utility. 
We also point out that previous results on the compatibility of PROP and PO were directly deduced from the compatibility of EF and PO. However, as an EF$\alpha$ allocation may not exist, we directly study PROP$\alpha$ in this paper.

The relations between the ``up to a fraction'' fairness notions and other well-known notions in the mixed goods setting are discussed in Section~\ref{sec:extension}.

\subsection{Related Work}
The study of fair allocation is extensive (see, e.g.,~\citet{brams1995envy,RobertsonWe98,Moulin19,Suksompong21,AmanatidisAzBi23} for a survey).
To capture fairness, various notions have been proposed for divisible and indivisible goods, including EF \citep{Foley67,varian1973equity} and PROP \citep{Steinhaus49} for divisible goods; EF1 \citep{LiptonMaMo04,Budish11} and PROP1 \citep{CFS17} for indivisible goods. 
There are also some notable fairness notions, e.g., {\em envy-freeness up to any good} (EFX) \citep{LiptonMaMo04,gourves2014near} and {\em maximin share} (MMS) \citep{Budish11}, etc. 

Recently, a stream of literature has focused on the fair allocation problem with a mixture of divisible and indivisible goods (mixed goods)~\citep{BeiLiLi20,BeiLiLu20,BhaskarSrVa21,NishimuraSu23,BeiLiLu23}. In particular, \citet{BeiLiLi20} initiated the fair division problem with mixed goods and proposed the fairness notion {\em envy-freeness for mixed goods} (EFM). Further, \citet{BeiLiLu20} and \citet{KawaseNiSu23} considered the fairness notions of MMS and {\em envy-freeness up to one good for mixed goods} (EF1M) in the mixed goods setting, respectively.
Note that, the ratio of approximate MMS allocation obtained in \citep{BeiLiLu20} is a monotonically increasing function determined by how agents value the divisible goods relative to their MMS values. On the other hand, our proposed indivisible ratio is determined by how an agent values the divisible goods relative to her value of all goods.
Furthermore, \citet{LiLiLu23} and \citet{LiLiLu24} examined EFM in conjunction with the issues of truthfulness and price of fairness, respectively. 
See a recent survey on the mixed fair division for more details~\citep{LiuLuSu23}.

In addition, several studies considered the interplay between fairness and efficiency for fair allocation~\citep{BKV18,BK19,FreemanSVX19,garg2021fair,aziz2020prop1,conf/ijcai/00010G21}.
Specifically, EF, EF1, and EF1M are compatible with Pareto optimality (PO) (i.e., a criterion of efficiency) via the maximum Nash welfare allocation in the divisible, indivisible, and mixed goods settings, respectively~\citep{SS19,CaragiannisKuMo19}.
It is worth noting that EFM is incompatible with PO while whether a weak version of EFM can be combined with PO is an open question in the mixed goods setting~\citep{BeiLiLi20}.

\section{Preliminaries} \label{sec:preliminaries}
Let $[k]$ denote the set $\{1,2,\ldots,k\}$ for any positive integer $k$.
We consider the mixed goods setting. 
Denote by $N = \{1, 2, \dots, n\}$ the set of $n$ agents, $M = \{o_1, o_2, \dots, o_m\}$ the set of $m$ indivisible goods, 
and $C = [0, 1]$ the set of heterogeneous divisible goods or a single \emph{cake}.\footnote{When there are more than one heterogeneous divisible goods, say $C = \{c_1, c_2, \ldots, c_\ell\}$, each cake can be represented by an interval $[\frac{j-1}{\ell}, \frac{j}{\ell})$, and thus the entire set of divisible goods can be regarded as a single cake $C = [0, 1)$. Later we assume agents' utility functions over $C$ are non-atomic, and hence the cake $[0,1)$ is equivalent to $[0, 1]$.}
Define $A = M \cup C$ to be the set of mixed goods.
An \emph{allocation} of the mixed goods is defined as $\mathcal{A} = (A_1, A_2, \dots, A_n)$ where $A_i = M_i \cup C_i$ is the \emph{bundle} allocated to agent $i$ subject to:
1) $C_i$ is a union of countably many intervals;
2) for any $i, j \in [n]$, $M_i \cap M_j = \emptyset$ and $C_i \cap C_j = \emptyset$;
3) $\bigcup_{i\in[n]} A_i = A$.
Each agent $i$ has a non-negative utility function $u_i(\cdot)$. Assume that each $u_i$ is additive over $A$ and integrable over $C$, that is, for any $M'\subseteq M$ and $C'\subseteq C$, $u_i(M'\cup C')=\sum_{o\in M'}u_i(o)+\int_{C'} u_i(x) dx$.
We also assume without loss of generality that agents' utilities are normalized to 1, i.e., $u_i(M \cup C) = 1$ for each $i \in N$.

We first show the classic fairness notions in the literature.

\begin{definition}[EF \& PROP]
    An allocation $\mathcal{A}$ is called
 \begin{itemize}
    \item \emph{envy-freeness (EF)} if for any agents $i, j \in N$, $u_i(A_i) \geq u_i(A_j)$;
    \item \emph{proportionality (PROP)} if for any agent $i \in N$, $u_i(A_i) \geq 1 / n$.
 \end{itemize}
\end{definition}

As mentioned before, for indivisible goods, relaxations of EF/PROP are commonly studied in the previous works.

\begin{definition}[EF1 \& PROP1]
    An allocation $\mathcal{A}$ is called
    \begin{itemize}
        \item \emph{envy-freeness up to one good (EF1)} if for any agents $i, j \in N$, there exists an indivisible good $o \in M_j$ such that $u_i(A_i) \geq u_i(A_j \setminus \{o\})$;
        \item \emph{proportionality up to one good (PROP1)} if for any agent $i \in N$, there exists an indivisible good $o \in M \backslash M_i$ such that $u_i(A_i) + u_i(o) \geq 1/n$.
    \end{itemize}
\end{definition}

However, as illustrated in the introduction, EF1 and PROP1 are rather weak in the mixed goods setting. In this paper, we introduce new ``up to a fraction'' fairness notions with the help of \emph{indivisibility ratio}.

\begin{definition}[Indivisibility Ratio]\label{def:alpha}
For each agent $i$, the \emph{indivisibility ratio} $\alpha_i$ is defined as $\alpha_i := \frac{u_i(M)}{u_i(M) + u_i(C)}$.
\end{definition}

For each agent $i$, $\alpha_i$ is the ratio between the utility for all \emph{indivisible} goods and the utility for all goods. 
We point out that each agent has a {\em personalized} indivisibility ratio, allowing us to define the fairness with respect to each agent's perspective. 
Specifically, we introduce the following new fairness notions.

\begin{definition}[EF$\alpha$ \& PROP$\alpha$]
    An allocation $\mathcal{A}$ is called
    \begin{itemize}
        \item \emph{envy-freeness up to $\alpha$-fraction of one good (EF$\alpha$)} if for any agents $i, j \in N$, there exists an indivisible good $o \in M_j$ such that $u_i(A_i) \geq u_i(A_{j}) - \alpha_i \cdot u_i(o)$.
        \item \emph{proportionality up to $\alpha$-fraction of one good (PROP$\alpha$)} if for any agent $i \in N$, there exists an indivisible good $o \in M \backslash M_i$ such that $u_i(A_i) + \alpha_i \cdot u_i(o) \geq 1/n$.
    \end{itemize}
\end{definition}

It is easy to observe that when an agent has a higher utility for the cake, her indivisible ratio becomes smaller. This, in turn, implies that she is more likely to receive an allocation closer to EF/PROP under the EF$\alpha$/PROP$\alpha$ criteria.
One can easily check that when good is only the cake, EF$\alpha$ (resp., PROP$\alpha$) reduces to EF (resp., PROP); when goods are all indivisible, EF$\alpha$ (resp., PROP$\alpha$) reduces to EF1 (resp., PROP1). We can also observe that EF$\alpha$ implies PROP$\alpha$.

As we will show later, an EF$\alpha$ allocation may not exist and a PROP$\alpha$ allocation always exists. 
For a better understanding of what EF$\alpha$
needs to be relaxed and PROP$\alpha$ can be strengthened, we next introduce the generalizations of EF$\alpha$ and PROP$\alpha$.

\begin{definition}[EF$f(\alpha)$ \& PROP$f(\alpha)$]
    An allocation $\mathcal{A}$ is 
    \begin{itemize}
        \item \emph{envy-freeness up to one $f(\alpha)$-fraction of good (EF$f(\alpha)$)} if for any agents $i, j \in N$, there exists an indivisible good $o \in M_j$ such that $u_i(A_i) \geq u_i(A_{j}) - f(\alpha_i) \cdot u_i(o)$.
        \item \emph{proportionality up to one $f(\alpha)$-fraction of good (PROP$f(\alpha)$)} if for any agent $i \in N$, there exists an indivisible good $o \in M \backslash M_i$ such that $u_i(A_i) + f(\alpha_i) \cdot u_i(o) \geq 1/n$.
    \end{itemize}
\end{definition}

When $f(\alpha) = \alpha$, the above notions degenerate to \efa~and \propa. In this paper, we focus on the linear function form $f(\alpha) = g(n) \cdot \alpha$, where $g(n)$ is a function of the number of agents. One can obtain stronger (resp., weaker) fairness requirements by making $g(n)$ smaller (resp., larger). 

We also consider the efficiency of the allocations.

\begin{definition}[PO]
    An allocation $\mathcal{A}$ is said to satisfy \emph{Pareto optimality (PO)} if there is no allocation $\mathcal{A'}$ that Pareto-dominates $\mathcal{A}$, i.e., $u_i(A'_i) \geq u_i(A_i)$ for all agents $i \in N$ and $u_i(A'_i) > u_i(A_i)$ for some agents $i \in N$.
\end{definition}

\begin{definition}[MNW]
    An allocation $\mathcal{A}$ is a maximum Nash welfare (MNW) allocation if the number of agents with positive utility is maximized, and subject to that, the product of the positive utilities ($\prod_{i\in [n]:u_i(A_i) > 0} u_i(A_i)$) is maximized.
\end{definition}

Finally, we utilize the Robertson-Webb (RW) query model \citep{RobertsonWe98} for accessing agents' utility functions over the cake. The RW model allows algorithms to query the agents with the following two methods: 1) an {\em evaluation} query on $[x,y]$ for agent $i$ returns $u_i([x,y])$, and 2) a {\em cut} query of $\beta$ for agent $i$ from $x$ returns the leftmost point $y$ such that $u_i([x,y])=\beta$, or reports no such $y$ exists.
In this paper, we assume each RW query using $O(1)$ time.

\section{Envy-freeness up to a Fractional Good}\label{sec:efalpha}
In this section, we focus on envy-freeness up to a fractional good, i.e., EF$\alpha$ and EF$f(\alpha)$.
We first present that for two agents, an EF$\alpha$ + PO allocation always exists and an EF$\alpha$ allocation can be found in polynomial time.
Then, we proceed to consider the case with $n \geq 3$ agents and show that there does not exist EF$(\frac{n^2}{4(n-1)}-\varepsilon)\alpha$ allocations for any $\varepsilon > 0$. We then explore the best fairness guarantee under EF$f(\alpha)$. In particular, we find that an EF$n\alpha$ allocation always exists which is tight up to a constant factor.
When agents have identical utility, we further show that $f(\alpha)=\frac{n^2}{4(n-1)}\alpha$ is the exact fraction we can guarantee for EF$f(\alpha)$.

\subsection{Two Agents} \label{subsec:efaplha2agents}
In this part, we first make use of the polynomial-time algorithm for finding an EFM allocation with two agents in~\citep{BeiLiLi20} to provide the existence of EF$\alpha$ allocations for two agents.

\begin{theorem}\label{thm:efalpha2agentspoly}
   When $n=2$, an EF$\alpha$ allocation exists and can be found in polynomial time.
\end{theorem}

\begin{proof}
The polynomial-time algorithm for finding an EFM allocation with two agents is also capable of finding an EF$\alpha$ allocation with two agents (\cite{BeiLiLi20}): ``We begin with an EF1 allocation $(M_1, M_2)$ of all indivisible goods. Assume without loss of generality that $u_1(M_1) \geq u_1(M_2)$. Next agent 1 adds the cake into $M_1$ and $M_2$ so that the two bundles are as close to each other as possible. Note that if $u_1(M_1) > u_1(M_2 \cup C)$, agent 1 would add all cake to $M_2$. If $u_1(M_1) \leq u_1(M_2 \cup C)$, agent 1 has a way to make the two bundles equal.
We then give agent 2 her preferred bundle and leave to agent 1 the remaining bundle.''
It is easy to see that we only need to analyze the case when $u_1(M_1) > u_1(M_2 \cup C)$ holds and agent 1 gets $M_2\cup C$. Since there exists some good $g$ in $M_1$ such that $u_1(M_2) \geq u_1(M_1 \setminus \{g\})$, we have
\begin{align*}
u_1(M_2 \cup C) & = u_1(M_2) + (1-\alpha_1) \\
&\geq u_1(M_1) - u_1(\{g\}) + (1-\alpha_1)u_1(\{g\})\\
&= u_1(M_1)  - \alpha_1 \cdot u_1(\{g\}), 
\end{align*}
which completes the proof. 
\end{proof}

We now proceed to show the compatibility of EF$\alpha$ and PO for two agents.
In particular, we first consider an allocation obtained via a variant of the cut-and-choose procedure: the first agent partitions the goods into two bundles $A_x$ and $A_y$ as equal as possible (assume $\max\{u_2(A_x),u_2(A_y)\}$ is maximized if multiple such partitions exist), and the second agent chooses first. Such an allocation can be utilized to return an EF$\alpha$ and PO allocation through Pareto improvements, which leads to the following theorem.

\begin{theorem}\label{thm:efalphaPO2agents}
    When $n=2$, an EF$\alpha$ and PO allocation always exists .
\end{theorem}

\begin{proof}
Let agent 1 partition the goods into two bundles $A_x$ and $A_y$ such that her values for the bundles are as equal as possible. Formally, let $u_1(A_x) = x$ and $u_1(A_y) = y$, and assume that $x \leq y$. We consider the partition that minimize $y - x$. If there are multiple such partitions, we only focus on the  partition for which $u_2(A_y)$ is maximized. Let agent 2 choose the bundle that she prefers, giving the other bundle to agent 1. We will show that the partition can be utilized to return an EF$\alpha$ and PO allocation. 

If (1) agent 2 chooses $A_x$ or (2) agent 2 chooses $A_y$ and $x = y$, then the final allocation is EF and hence PROP. Any Pareto improvement maintains PROP which proves the existence of PROP (thus EF and EF$\alpha$) and PO. 

The remaining case is that agent 2 chooses $A_y$ and $x < y$.
We consider the Pareto optimality of the current allocation $(A_x, A_y)$. Suppose to the contrary that there is another allocation $(A_x', A_y')$ that is Pareto-dominating $(A_x, A_y)$. 
If this allocation $(A_x', A_y')$ is PROP (and thus EF and EF$\alpha$), we can utilize it to obtain an EF$\alpha$ and PO allocation from the same analysis as above.
Otherwise, since $u_2(A_y')\ge u_2(A_y)\ge\frac12$, we have: either (1) $\frac12>u_1(A_x')>u_1(A_x)$ or (2) $\frac12>u_1(A_x')=u_1(A_x)$ and $u_2(A_y')>u_2(A_y)$.
Both cases lead to a contradiction with the construction of the partition $(A_x,A_y)$.
It suffices to show $(A_x,A_y)$ is an \efa~allocation.
As agent 2 is clearly EF, we only need to show that agent 1 in this case is EF$\alpha$. 

Assume to the contrary that we have $x < y - \alpha_1 \cdot u_1(g)$ for all $g \in A_y$. Thus we have $x - y < -\alpha_1 \cdot u_1(g)$ and $y - x > \alpha_1 \cdot u_1(g)$. 
Observing that $A_y$ does not contain any pieces of cake with positive value, otherwise moving some pieces of cake to $A_x$ would create a more equal partition, which leads to a contradiction.
Moreover, by definition of $\alpha_1$, we have $u_1(C) = \frac{1-\alpha_1}{\alpha_1} u_1(M) \geq (1-\alpha_1) u_1(M)$. This implies that we can move some indivisible good $g \in A_y$ to $A_x$ and move some pieces of cake with value $(1 - \alpha_1) \cdot u_1(g)$ to $A_y$. Denote the resulting allocations as $A_x'$ and $A_y'$. Now, agent 1 has value $x' = x + u_1(g) - (1 - \alpha_1) \cdot u_1(g) = x + \alpha_1 \cdot u_1(g)$ for $A_x'$ and $y' = y - u_1(g) + (1 - \alpha_1) \cdot u_1(g) = y - \alpha_1 \cdot u_1(g)$ for $A_y'$. If we still have $x' \leq y'$, it is clear that $y' - x' < y - x$; otherwise, we have $x' - y' = x - y + 2\alpha_1 \cdot u_1(g) < \alpha_1 \cdot u_1(g) < y - x$. It means that in either case we are able to create a more equal partition, a contradiction. This completes the proof.
\end{proof}

\subsection{General Number of Agents} \label{subsec:efalpha3agents}
We move on to consider the case of a general number of agents $n \geq 3$ in this part.
When the resources to be allocated contain only divisible or indivisible goods, EF$\alpha$ allocations always exist.
However, when the goods are mixed, we show that EF$\alpha$ allocations fail to exist even when there is only one homogeneous cake\footnote{We call a cake \emph{homogeneous} if the utility over a subset of cake $C$ depends only on the length of this subset, i.e., for each $i\in[n]$ and any $C'\subseteq C$, $u_i(C')=\frac{|C'|}{|C|} \cdot u_i(C)$, where $|C'|$ and $|C|$ represents the length of $C'$ and $C$, respectively. 
} 
and one indivisible good.

\begin{theorem}\label{thm:efalpha3agents}
    For $n \geq 3$ agents, an EF$\alpha$ allocation does not always exist. Specifically, for any $\varepsilon > 0$, an EF$(\frac{n^2}{4(n-1)}-\varepsilon)\alpha$ allocation does not always exist.
\end{theorem}

\begin{proof}
The proof is derived from the following counterexample where we have $n$ identical agents, one indivisible good $o$, and one homogeneous cake $C$.
    \begin{center}
    \begin{tabular}{cccc}
        \toprule
        & $o$ & $C$ & $\alpha$ \\
        \midrule
        $u_i(\cdot), \forall i \in [n]$ & $\frac{2}{n}$ & $\frac{n-2}{n}$ & $\frac{2}{n}$ \\
        \bottomrule
    \end{tabular}
    \end{center}
Suppose the indivisible good $o$ is allocated to agent $n$. Then there must exist one agent $i \in [n-1]$ such that $u_i(A_i) \le \frac{n-2}{n(n-1)}$. For this agent $i$,
\begin{equation}
    u_i(A_n) - \alpha \cdot u_i(o) = \frac{2}{n} - (\frac{2}{n})^2 = \frac{2(n-2)}{n^2}.
\end{equation}
When $n \geq 3$, we have $\frac{2}{n} > \frac{1}{n-1}$, and thus $\frac{2(n-2)}{n^2} > \frac{n-2}{n(n-1)}$ which implies that agent $i$ fails to achieve EF$\alpha$. 

The above analysis can be tighter, and we have 
    \begin{eqnarray*}
        && u_i(A_n) - (\frac{n^2}{4(n-1)}-\varepsilon)\alpha \cdot u_i(o) \\
        &=& \frac{n-2}{n(n-1)} + \frac{4\varepsilon}{n^2} > u_i(A_i),
    \end{eqnarray*}
which implies agent $i$ fails to achieve EF$(\frac{n^2}{4(n-1)}-\varepsilon)\alpha$.
\end{proof}

On the positive side, we will show in Section~\ref{sec:extension} that EFM implies EF$n\alpha$ (Theorem~\ref{thm:efm implies efna}) which means that an EF$n\alpha$ allocation can be derived by using the polynomial algorithm for EFM in \citep{BeiLiLi20}.
This also implies that when $n \geq 3$, the best fairness guarantee under \effa~notion would be EF$\Theta(n)\alpha$.
Further, when all agents are identical, we show that the exact fairness guarantee under \effa~notion is EF$\frac{n^2}{4(n-1)}\alpha$.

\begin{theorem} \label{thm:identical efn/4a}
    When agents have identical utility functions, an EF$\frac{n^2}{4(n-1)}\alpha$ allocation always exist.
\end{theorem}

\begin{proof}
    As described after the statement of Theorem~\ref{thm:identical efn/4a}, we adopt an algorithm which finds an EF1 allocation for indivisible goods first and then utilizes the water-filling procedure to allocate the cake, where the cake will be always allocated to the agents with the smallest utility.
    
    Based on the form $f(\alpha)=g(n)\cdot\alpha$, we then prove that for an arbitrary integer $n\ge 3$, via this algorithm, each instance that admits a function $g(n)$ can be modified to an instance with only one indivisible good and a cake that also admits a function lower bounded by $g(n)$.

    Under the instance $I_1$ that admits a function $g(n)$, we assume the envy of $g(n)\cdot \alpha$ occurs from agent $i$ to agent $j$.
    Since when $j$ is fixed, the smaller $u(i)$ can lead to a larger $g(n)$, we can assume  the bundle $i$ has the smallest utility among all bundles, which always receives divisible goods during the whole water-filling process and has exactly the utility at the water level (the utility of all agents receiving some cake).

    Since the allocation is EF1 before allocating the cake, the difference between the utilities of bundle $j$ and bundle $i$ before allocating the divisible goods is upper bounded by some good $g\in M_j$. We then assume $g'$ is the good with the largest utility in $M_j$, which is the corresponding $o$ used in the condition of EF$f(\alpha)$.
    We can then compare to the instance $I_2$ where only one indivisible good with utility $g'$ is given to agent $j$.
    In both instances, the $g'$s corresponding to the $o$ used in the condition of EF$f(\alpha)$ are the same but the initial difference between the utilities of agent $j$ and $i$ is larger under $I_2$, which is from a value weakly less than $u(g)\le u(g')$ to exactly $u(g')$.
    Further, the divisible goods are required to be allocated to exactly $n-1$ agents from the beginning in $I_2$, which can makes a lower water level and a larger difference after terminating the water-filling procedure.
    With the decrease of $\alpha$ since there exists only one indivisible good $g'$ in $I_2$, instance $I_2$ admits a larger $g(n)$ than that in $I_1$.

    We can then focus on the case where there is only one indivisible good and a cake.
    We assume the value over the indivisible good is $x$ and the value over the cake is $1-x$, where $\alpha$ here is exactly $x$.
    Without loss of generality, we assume $x\ge \frac{1-x}{n-1}$ otherwise we can directly reach a PROP allocation.
    We then need to find an $x$ which minimizes the $g(n)$ in the inequality $x-g(n)\cdot x^2\le \frac{1-x}{n-1}$, which is the condition for EF$f(\alpha)$.
    By some calculus, the optimal $x$ is exactly $\frac2n$, corresponding to the counterexample in Theorem~\ref{thm:efalpha3agents}.
\end{proof}

However, such an idea is not applied to the non-identical agents setting, where the presence of multiple indivisible goods may complicate the envy graph and prevent us from reducing it to a case with only one indivisible good. 

\section{Existence and Computation of PROP$f(\alpha)$} \label{sec:propalpha}
In this section, we focus on the proportionality up to a fractional good (i.e., PROP$\alpha$ and PROP$f(\alpha)$). 
We first prove by presenting a polynomial algorithm that a PROP$\alpha$ allocation always exists in the mixed good setting. Subsequently, we consider the \propfa~notion and show a lower bound of $f(\alpha)$, giving an asymptotically tight characterization for the existence of PROP$f(\alpha)$.

\subsection{The Algorithm}

The complete algorithm for finding a PROP$\alpha$ allocation is shown in Algorithm~\ref{alg:propalpha}.
Conceptually, Algorithm~\ref{alg:propalpha} performs the ``moving-knife'' procedure on indivisible goods and the cake separately through $n - 1$ rounds (Steps \ref{propALG-iterBegin}-\ref{propALG-iterEnd}). In each round, one agent is allocated with a bundle that yields a utility that achieves PROP$\alpha$. The bundle is firstly filled by indivisible goods (Step \ref{propALG-lastInd}) until there exists an indivisible good $o$ such that after adding $o$ to the bundle, some agent $j$ will be satisfied with the bundle. Then, depending on whether we can make some agent satisfied by only adding some cake to the bundle, we execute either Case 1 (Steps \ref{propALG-case1Begin}-\ref{propALG-case1End}) or Case 2 (Steps \ref{propALG-case2Begin}-\ref{propALG-case2End}).

\begin{itemize}
    \item Case 1: when the cake is not large enough to make any agent satisfied, we simply add the indivisible good $o$ to the bundle and allocate it to agent $j$.
    \item Case 2: we first find out the minimum required piece of cake for each agent that will make her satisfied (Step \ref{propALG-RW}). This step can be implemented through the RW model. Note that the condition of entering the second case ensures that at least one agent can be satisfied by adding some piece of the cake. we then find the optimal agent $i^*$ that requires the minimum piece of cake among all agents and allocate the piece of cake together with the bundle of indivisible goods to her.
\end{itemize}

After allocating the bundles to $n - 1$ agents, we give all the remaining goods to the last agent (Step \ref{propALG:lastAgent}).
Assume, w.l.o.g., that agents $1, 2, \dots, n$ receive their bundles in order.

We remark that when all goods are indivisible, our algorithm omits Case 2 and thus degenerates to the well-known bag-filling procedure. When the whole good is a divisible cake, the algorithm only executes Case 2 and then becomes the classical moving-knife algorithm. Though Algorithm~\ref{alg:propalpha} is a natural extension of the algorithms in the divisible goods setting and the indivisible goods settings, we claim that the analysis is non-trivial as shown in the next subsection.

\begin{algorithm}[t]
    \caption{Finding a PROP$\alpha$ allocation}
    \label{alg:propalpha}
    \KwData{Agents $N$, indivisible goods $M$ and cake $C$}
    \KwResult{A PROP$\alpha$ allocation $(A_1, A_2, \dots, A_n)$}
    $\hat{M} \leftarrow M, \hat{C} \leftarrow C$\;
    \While {$|N| \geq 2$} { \label{propALG-iterBegin}
        $B \leftarrow \emptyset$\;
        Add one indivisible good in $\hat{M}$ at a time to $B$ until adding the next indivisible good $o$ will cause $u_j(B \cup \{o\}) \geq 1 / n - \alpha_j \cdot u_j(g)$ for some agent $j$ and some good $g \in M \setminus (B \cup \{o\})$, or $M \setminus (B \cup \{o\}) = \emptyset$\; \label{propALG-lastInd}
        \eIf{$~\forall i \in N$ and $g \in M \setminus B$, $u_i(B \cup \hat{C}) < 1 / n - \alpha_i \cdot u_i(g)$} { \label{propALG-case1Begin}
            // Case 1: allocate with only indivisible goods\;
            $A_{j} \leftarrow B \cup \{o\}$\;
            $N \leftarrow N \setminus \{j\}, \hat{M} \leftarrow \hat{M} \setminus (B \cup \{o\})$\; \label{propALG-case1End}
        }{ \label{propALG-case2Begin}
            // Case 2: allocate with cake\; 
            Suppose now $\hat{C} = [a, b]$. For all $i \in N$, if $u_i(B\cup [a,b]) \geq 1 / n - \alpha_i \cdot u_i(g)$, let $x_i$ be the leftmost point such that $u_i(B \cup [a, x_i) ) = 1 / n - \alpha_i \cdot u_i(g)$ for some good $g \in M \setminus B$; otherwise, let $x_i = b$\; \label{propALG-RW}
            $i^* \leftarrow \arg\min_{i \in N} x_i$\;
            $A_{i^*} \leftarrow B \cup [a, x_{i^*})$\;
            $N \leftarrow N \setminus \{i^*\}, \hat{M} \leftarrow \hat{M} \setminus B, \hat{C} \leftarrow \hat{C} \setminus [a, x_{i^*})$\; \label{propALG-case2End} \label{propALG-iterEnd}
        } 
    }
    Give all the remaining goods to the last agent\; \label{propALG:lastAgent}
    \Return $(A_1, A_2, \dots, A_n)$\;
\end{algorithm}

\subsection{Analysis}

Our main result for the existence and computation of PROP$\alpha$ allocations is as follows.
\begin{theorem}\label{thm:propalpha}
    For any number of agents, Algorithm~\ref{alg:propalpha} returns a PROP$\alpha$ allocation in polynomial time.
\end{theorem}

To prove Theorem~\ref{thm:propalpha}, we will utilize some useful concepts and facts.
Let $k$ be the first agent that receives her bundle with only indivisible goods in Steps~\ref{propALG-case1Begin} to~\ref{propALG-case1End}.
In other words, agents $1, \dots, k-1$ are assigned their bundles with some pieces of cake in Steps~\ref{propALG-case2Begin} to~\ref{propALG-case2End}.
Moreover, for distinction, we set $k \leftarrow n$ if all the first $n-1$ agents are assigned in Steps~\ref{propALG-case2Begin} to~\ref{propALG-case2End}. 
The agent $k$ plays an important role in bounding the ratio $\alpha_i$ for the agent after $k$. The relation is shown below.

\begin{claim}\label{clm:alpha}
$\alpha_i \geq \frac{n - k}{n}$ for each agent $i > k$.
\end{claim}

\begin{proof}
Fix an agent $i > k$. 
According to the definition of $k$, each agent $j$ in $\{1, \dots, k-1\}$ will be assigned a bundle with some pieces of cake $C_j$. It is clear to see that agent $i$'s valuation on each $C_j$ is at most $1 /n$ since otherwise agent $i$ will be an agent that receives a bundle before agent $k$.

Next, we focus on the remaining cake. Note that if the condition in Step~\ref{propALG-case1Begin} is true, i.e., $u_i(B \cup \hat{C}) < 1 / n - \alpha_i \cdot u_i(g)$ for all goods $g \in M \setminus B$, this also means that the remaining cake $\hat{C}$ worth at most $1 / n$ to agent $i$.

Combined the above arguments, we have $u_i(C) \leq k / n$ which completes the proof.
\end{proof}

Let $o_i$ be the indivisible good $o$ as defined in Step~\ref{propALG-lastInd} for the $i$-th iteration of the while-loop.
For each agent $j$, we define $g_{ij} = \arg \max_{g \in (M \setminus M_i) \cup {o_i}} u_j(g)$. Intuitively, $g_{ij}$ is the good that, when agent $i$ conducts the ``moving-knife'' procedure (and finally obtains $A_i$), for agent $j$, the most valuable good besides the bundle $A_i$ without $o_i$. The next claim makes a connection between the above two definitions.

\begin{claim}\label{clm:o}
$u_j(g_{ij}) \geq u_j(o_p)$ for any agents $i, j, p$ with $i \neq n$ and $o_p \neq o_{i-1}$.
\end{claim}

\begin{proof}
For the first case when $p>i$, because $o_p$ exists in round $p$, we have $o_p\notin A_i$ and then verify the case from the definitions of $g_{ij}$.

When $p<i$, since $o_p\neq o_{i-1}$, $o_p$ is not observed at round $i$ and we can prove the statement by the definitions of $g_{ij}$.

When $p=i$, the statement obviously holds.
\end{proof}

We are ready to prove the following lemma.

\begin{lemma} \label{lem:enough}
Before the $j$-th iteration, the remaining goods are enough for $j$ to achieve PROP$\alpha$. Specifically, we have
$$u_j (\hat{M} \cup \hat{C}) \geq (n-j+1)(\frac{1}{n} - \alpha_j \cdot u_j(g_{jj})),$$
where $\hat{M}\cup\hat{C}$ represents the remaining goods just before the $j$-th iteration.
\end{lemma}

\begin{proof}
We first look at the case with $j \leq k$. By definition of $k$, each agent in $\{1, \dots, k - 1\}$ is assigned her bundle in Steps~\ref{propALG-case2Begin} to~\ref{propALG-case2End}. In other words, from the viewpoint of $j$, the value of each bundle $A_i$ with $i \in \{1, \dots, j - 1\}$ is less than $1/n - \alpha_j \cdot u_j(g) \leq 1/n$ for some $g \not \in A_i$ since, otherwise, agent $j$ will be assigned before agent $i$.
This means that, agent $j$'s value on the remaining goods is at least $$1 - \frac{j-1}{n} = \frac{n-j+1}{n} \geq (n-j+1)(\frac{1}{n} - \alpha_j \cdot u_j(g_{jj})).$$ 
We note that $k = n$ belongs to this case as the analysis only focus on the fact that each agent in $\{1, \dots, k - 1\}$ is assigned her bundle in Steps~\ref{propALG-case2Begin} to~\ref{propALG-case2End}.

We next consider the case with $j > k$. Similar to the above case, we can show that agent $j$'s value on each bundle $A_i$ with $i \in \{1, \dots, k - 1\}$ is weakly less than $1/n - \alpha_j \cdot u_j(g_{ij})$. We then focus on the $k$-th assignment to the $(j-1)$-st assignment which can either be implemented in Steps~\ref{propALG-case1Begin} to~\ref{propALG-case1End} or in Steps~\ref{propALG-case2Begin} to~\ref{propALG-case2End}. 
If the $i$-th assignment, where $k \leq i \leq j-1$, is executed by Steps~\ref{propALG-case2Begin} to~\ref{propALG-case2End}, we can similarly obtain that $u_j (A_i) \leq 1/n - \alpha_j \cdot u_j(g_{ij})$. For the situation with Steps~\ref{propALG-case1Begin} to~\ref{propALG-case1End}, 
we can also have $u_j (A_i) \leq 1/n - \alpha_j \cdot u_j(g_{ij}) + u_j(o_i)$. 
Now we are ready to compute agent $j$'s value over the remaining goods. Our discussion is divided into two parts: (a) $\arg \max_{k \leq p \leq j-1} u_j(o_p) = \{o_{j-1}\}$ and $o_{j-1} \in A_j$; (b) other situations, i.e., either (i) $\arg \max_{k \leq p \leq j-1} u_j(o_p) \neq \{o_{j-1}\}$ or (ii) $\arg \max_{k \leq p \leq j-1} u_j(o_p) = \{o_{j-1}\}$ but $o_{j-1} \not \in A_j$. 

We first consider Part (b), we can get:
\begin{align*}
& \quad u_j(\hat{M} \cup \hat{C}) \\
& = 1 - \sum_{1 \leq i \leq j-1} u_j(A_i) \\
&\geq 1 - \sum_{1 \leq i \leq k-1} \left(\frac{1}{n} - \alpha_j \cdot u_j(g_{ij})\right)  \\
& \quad \quad \quad - \sum_{k \leq i \leq j-1} \left(\frac{1}{n} - \alpha_j \cdot u_j(g_{ij}) + u_j(o_i)\right) \\
&= \frac{n-j+1}{n} + \sum_{1 \leq i \leq k-1} (\alpha_j \cdot u_j(g_{ij})) \\
& \quad \quad \quad + \sum_{k \leq i \leq j-1} (\alpha_j \cdot u_j(g_{ij}) - u_j(o_i)) \\
&\geq \frac{n-j+1}{n} + \alpha_j \sum_{1 \leq i \leq k-1} \max_{k \leq p \leq j-1} u_j(o_p) \\
& \quad \quad \quad  + (\alpha_j - 1) \sum_{k \leq i \leq j-1} \max_{k \leq p \leq j-1} u_j(o_p)
\end{align*}
Here, for the second inequality, we can first observe that $u_j(g_{ij})\ge \max_{k \leq p \leq j-1} u_j(o_p)$ for each $1\le i\le k-1$ from Claim~\ref{clm:o}.
For the third term in this inequality, we achieve this by presenting the fact that $\alpha_j \cdot u_j(g_{ij}) - u_j(o_i)\ge  (\alpha_j - 1)\max_{k \leq p \leq j-1} u_j(o_p)$ is satisfied for each $k \leq i \leq j-1$.
We can utilize $u_j(o_i)\le \max_{k \leq p \leq j-1} u_j(o_p)$ under the case that $u_j(g_{ij})\ge \max_{k \leq p \leq j-1} u_j(o_p)$ and apply $u_j(o_i)\le u_j(g_{ij})$ under the case that $u_j(g_{ij})<\max_{k \leq p \leq j-1} u_j(o_p)$ to show this. 
By rearranging the terms, this is equivalent to:
\begin{align*}
&\frac{n-j+1}{n} + \left(\alpha_j (k-1) + (\alpha_j - 1) (j-k) \right)\max_{k \leq p \leq j-1} u_j(o_p)\\
\geq& \frac{n-j+1}{n} + \left( \frac{(n-k)(k-1) - k(j-k)}{n} \right) \max_{k \leq p \leq j-1} u_j(o_p)\\
=& \frac{n-j+1}{n} + \left( \frac{-(n-k) + k (n-j)}{n} \right) \max_{k \leq p \leq j-1} u_j(o_p)\\
\geq& \frac{n-j+1}{n}  -\left( \frac{(n-j+1)(n-k)}{n} \right) \max_{k \leq p \leq j-1} u_j(o_p)\\
\geq& (n-j+1) \left(\frac{1}{n} - \alpha_j \cdot \max_{k \leq p \leq j-1} u_j(o_p)\right)\\
\geq& (n-j+1)\left( \frac{1}{n} - \alpha_j \cdot u_j(g_{jj})\right).
\end{align*}
The first and third inequalities are from Claim~\ref{clm:alpha}, while the fourth inequality is from the assumption of Part (b).
From the assumptions of Part (b), we can conclude that $\arg \max_{k \leq p \leq j-1} u_j(o_p)$ is not in the bundle $A_j$ so $u_j(g_{jj}) \geq \max_{k \leq p \leq j-1} u_j(o_p)$.

We next consider Part (a). 
We first define a term $r$ which is the largest value between $k$ and $j-1$ such that $o_r\neq o_{j-1}$. Since $o_k\in A_k$ which cannot be $o_{j-1}$, such $r$ must exist.
We observe that each agent $i$ from $r+1$ to $j-1$ gets her bundle with mixed goods in Steps~\ref{propALG-case2Begin} to~\ref{propALG-case2End} since $o_{i}=o_{j-1} \in A_j$. The analysis of Part (a) is similar to the one of Part (b) with some mild modifications.
\begin{align*}
& \quad u_j(M \cup C) \\
& = 1 - \sum_{1 \leq i \leq j-1} u_j(A_i) \\
&\geq 1 - \sum_{1 \leq i \leq k-1} \left(\frac{1}{n} - \alpha_j \cdot u_j(g_{ij})\right)\\
& \quad \quad \quad  -\sum_{k \leq i \leq r} \left(\frac{1}{n} - \alpha_j \cdot u_j(g_{ij}) + u_j(o_i)\right) \\
& \quad \quad \quad  - \sum_{r+1\le i\le j-1}\left(\frac{1}{n} - \alpha_j \cdot u_j(g_{ij})\right) \\
&\geq \frac{n-j+1}{n} + \sum_{1 \leq i \leq k-1 \atop or~r + 1 \leq i \leq k - 1} (\alpha_j \cdot u_j(g_{ij})) \\
& \quad \quad \quad  + \sum_{k \leq i \leq r} (\alpha_j \cdot u_j(g_{ij}) - u_j(o_i)) \\
&\geq \frac{n-j+1}{n} + \alpha_j \sum_{1 \leq i \leq k-1 \atop or~r + 1 \leq i \leq k - 1} \max_{k \leq p \leq r} u_j(o_p) \\
& \quad \quad \quad  + (\alpha_j - 1) \sum_{k \leq i \leq r} \max_{k \leq p \leq r} u_j(o_p) \\
&\geq (n-j+1)\left( \frac{1}{n} - \alpha_j \cdot u_j(g_{jj})\right),
\end{align*}
where the inequalities can be derived based on the same reason as in Part (b).
The proof of Lemma~\ref{lem:enough} is thus complete. 
\end{proof}

We now turn our attention to the proof of Theorem~\ref{thm:propalpha}.
\begin{proof}[Proof of Theorem~\ref{thm:propalpha}]
It is clear that all the goods are allocated after Step~\ref{propALG:lastAgent} of Algorithm~\ref{alg:propalpha}. We next consider the correctness of the algorithm. By Lemma~\ref{lem:enough}, we know that each iteration of the while-loop of Algorithm~\ref{alg:propalpha} is well-defined, which means that each agent in $\{1, \dots, n-1\}$ achieves her PROP$\alpha$, either in Steps~\ref{propALG-case1Begin} to~\ref{propALG-case1End} or in Steps~\ref{propALG-case2Begin} to~\ref{propALG-case2End}. For the last agent, Lemma~\ref{lem:enough} also implies that the PROP$\alpha$ of agent $n$ can be satisfied.

Since each step in this algorithm can be executed in polynomial time and the total number of the while loop at Step~\ref{propALG-iterBegin} is exactly $n-1$, we can conclude the polynomial running time of the algorithm, which completes the proof.
\end{proof}

\subsection{Impossibility Result}
From Theorem~\ref{thm:propalpha}, we know that a PROP$\alpha$ allocation always exists.
In the following, we give a lower bound on $f(\alpha)$ such that PROP$f(\alpha)$ allocation is not guaranteed to exist.

\begin{theorem} \label{thm:propa negative result}
    For any $\varepsilon > 0$, a PROP$(\frac{n-1}{n}-\varepsilon)\alpha$ allocation does not always exist.
\end{theorem}

\begin{proof}
    Without loss of generality, we assume $\varepsilon\le\frac25$, otherwise we can choose the corresponding instance for $\varepsilon=\frac25$ to prove this.
    Let $x = \frac{\varepsilon}{n-1}$. We consider the following instance with $n$ identical agents, $n - 1$ indivisible goods, and one cake.
    \begin{center}
    \begin{tabular}{cccc}
        \toprule
        & $|M| = n - 1$ & $C$ & $\alpha$ \\
        \midrule
        $u_i(\cdot), \forall i \in [n]$ & $\frac{1-x}{n-1},\forall o \in M$ & $x$ & $1-x$ \\
        \bottomrule
    \end{tabular}
    \end{center}
    There exists an agent $i$ who obtains no indivisible goods. Thus $u_i(A_i) \leq u_i(C) = x$. For this agent and any good $g \in M$, we have
    \begin{eqnarray*}
        && u_i(A_i) + (\frac{n-1}{n}-\varepsilon)\alpha_i \cdot u_i(g) \\
        &\leq& x + (\frac{1}{n} -x) \cdot (1-x)^2 \\
        &=& \frac{1}{n} - x(\frac{2}{n}-\frac{2n+1}{n}x + x^2) \\
        &=&\frac1n-x(\frac2n-\frac2n\frac{n+1/2}{n-1}\varepsilon+x^2)\\
        &<& \frac{1}{n}.
    \end{eqnarray*}
    The last inequality is because $\frac{(n+\frac12)\varepsilon}{n-1}\le 1$ from $n\ge 2$ and $\varepsilon\le \frac25$.
    Therefore, agent $i$ fails to achieve PROP$(\frac{n-1}{n}-\varepsilon)\alpha$.
\end{proof}

This impossibility result with Theorem~\ref{thm:propalpha} also implies that we obtain an asymptotically tight characterization for the existence of PROP$f(\alpha)$.

\section{Compatibility of PROP$\alpha$ and PO}\label{sec:propa+po}
In this section, we show the compatibility of PROP$\alpha$ and PO. In particular, we prove that any allocation that maximizes Nash welfare (which is trivially PO) must be PROP$\alpha$, and the implication regarding parameter $\alpha$ is tight. 
It is important to note that our approach is different from the conventional ones in \citep{CaragiannisKuMo19,KawaseNiSu23} since PROP$\alpha$ is not defined using pairwise comparisons between agents. 
PROP$\alpha$ requires an agent to compare her own bundle and the {\em union} of the goods allocated to all the other agents, and moreover, such a comparison, as relaxed by the indivisibility ratio times the value of the largest item outside of her bundle, is sensitive to the overall allocation. 
To overcome this complexity, we later present a monotone property of Nash welfare maximizing allocations (see \cref{cor:partition covers mnw move condition}) to quantify the effect of the reallocation of (fractional) goods.

\begin{theorem} \label{MNW is propa}
    Any MNW allocation satisfies \propa.
\end{theorem}

The proof of \cref{MNW is propa} relies on the following nice properties. The first one states that we can remove the goods that yield zero value to all agents without loss of generality.

\begin{observation}
    In an MNW allocation, if there exists a subset $p \subseteq A_i$ such that $u_i(p) = 0$, then for every $j \in [n]$ we have $u_j(p) = 0$. Therefore, we can simply remove $p$ from $A$. Formally, we assume with out loss of generality that 
    \begin{equation} \label{no zero utility item}
        \forall p \subseteq A_i, u_i(p) > 0.
    \end{equation}
\end{observation}

Base on the above assumption, we further show the following reduction that since any agent with zero utility in an MNW allocation must have at most $n$ positive-valued goods and no positive-valued cake, the agent achieves PROP$\alpha$ and we can focus only on other agents with positive utility.

\begin{lemma} \label{non-zero reduce}
    Suppose \cref{no zero utility item} holds. If \cref{MNW is propa} holds for every instances admitting an MNW allocation with no agent with zero utility, then \cref{MNW is propa} holds for every instances.
\end{lemma}

\begin{proof}
    We assume there exists such a $k$ that $u_i(A_i) > 0$ for $i \in [k]$ and $u_i(A_i) = 0$ for $i \in [n] \backslash [k]$. From \cref{no zero utility item} we know agents in $[n] \backslash [k]$ are allocated with no goods. Since MNW maximizes the product of utilities among agents in $[k]$, these agents achieve \propa~by the assumption of the lemma. 
    
    For a fixed agent $i \in [n] \backslash [k]$, we claim that $u_i(C) = 0$ because otherwise we can allocate some pieces of cake to agent $i$ to increase the number of agents with positive utility. Further, there are no more than $k$ goods in $M$ such that agent $i$ has positive utility on them. Otherwise, there is some agent in $[k]$ who has more than two such goods. By allocating one of them to agent $i$, the number of agents with positive utility increases, which leads to a contradiction. Therefore, since agent $i$ has positive utility on at most $k$ goods, there exists some goods $o \in M \backslash M_i$ such that $u_i(o) \geq \frac{u_i(A)}{n}$. Together with $\alpha_i = 1$, agent $i$ achieves \propa.
\end{proof}

\cref{non-zero reduce} allows us to assume that
\begin{equation} \label{wlog no zero agent}
    \forall i \in [n], u_i (A_i) > 0
\end{equation}
in the following proofs. In the next lemma, we show that if an agent $i$ does not envy some other agent, then the PROP$\alpha$ criterion can be reduced to the case when agent $i$ envies all other agents.

\begin{lemma} \label{EF reduce}
    Suppose \cref{no zero utility item} and \cref{wlog no zero agent} hold. If \cref{MNW is propa} holds for any agent that envies all other agents, \cref{MNW is propa} holds for every agents.
\end{lemma}

\begin{proof}
    We assume there exists such a $k$ that $u_1(A_1) < u_1(A_i)$ for $i \in [k] \backslash \{1\}$ and $u_1(A_1) \geq u_1(A_i)$ for $i \in [n] \backslash [k]$. 
    Since agent $1$ envies every agent in $[k] \backslash \{1\}$, we consider the instance containing only the first $k$ agents and the goods in the first $k$ bundles, we can then assume that there exists some $o \in \bigcup_{i=2}^{k} M_i$ such that
    \begin{equation} \label{EF reduce hypothesis}
        u_1(A_1) + \frac{\sum_{i=1}^{k}u_1 (M_i)}{\sum_{i=1}^{k} u_1 (A_i)} \cdot u_1(o) \geq \frac{\sum_{i=1}^{k} u_1 (A_i)}{k}.
    \end{equation}
    This inequality trivially holds when $k=1$.
    We then need to use \cref{EF reduce hypothesis} to prove that
    \begin{equation} \label{EF reduce result}
        \exists~o \in \bigcup_{i=2}^{n} M_i,~u_1(A_1) + \frac{u_1 (M)}{u_1 (A)} \cdot u_1(o) \geq \frac{u_1 (A)}{n}.
    \end{equation}

    Denote one feasible $o \in \bigcup_{i=2}^{k} M_i$ in \cref{EF reduce hypothesis} by $o^*$. Note that $o^* \in \bigcup_{i=2}^{n} M_i$. 
    Recall that $u_1(A_1) < u_1(A_i)$ for $i \in [k] \backslash \{1\}$ and $u_1(A_1) \geq u_1(A_i)$ for $i \in [n] \backslash [k]$, we have
    \begin{eqnarray*}
        && \frac{\sum_{i=k+1}^{n} u_1 (A_i)}{\sum_{i=1}^{k} u_1 (A_i)} \leq \frac{(n-k) u_1(A_1)}{k u_1(A_1)} = \frac{n-k}{k} \\
        &\Leftrightarrow& \frac{u_1 (A)}{\sum_{i=1}^{k} u_1 (A_i)} \leq \frac{n}{k} \\
        &\Leftrightarrow& \frac{k}{\sum_{i=1}^{k} u_1 (A_i)} \leq \frac{n}{u_1 (A)}.
    \end{eqnarray*}
    Then, we deduce from \cref{EF reduce hypothesis} that
    \begin{eqnarray*}
        \sum_{i=1}^{k} u_1 (A_i)
        &\leq& k u_1(A_1) + \frac{k \sum_{i=1}^{k}u_1 (M_i)}{\sum_{i=1}^{k} u_1 (A_i)} \cdot u_1(o^*) \\
        &\leq& k u_1(A_1) + \frac{n \sum_{i=1}^{k}u_1 (M_i)}{u_1 (A)} \cdot u_1(o^*) \\
        &\leq& k u_1(A_1) + \frac{n u_1 (M)}{u_1 (A)} \cdot u_1(o^*).
    \end{eqnarray*}
    Therefore, we have
    \begin{eqnarray*}
        u_1(A)
        &=&\sum_{i=1}^{n} u_1 (A_i) \\
        &\leq& k u_1(A_1) + \frac{n u_1 (M)}{u_1 (A)} \cdot u_1(o^*) + \sum_{i=k+1}^{n} u_1 (A_i) \\
        &\leq& k u_1(A_1) + \frac{n u_1 (M)}{u_1 (A)} \cdot u_1(o^*) + (n-k) u_1(A_1) \\
        &=& n u_1(A_1) + \frac{n u_1 (M)}{u_1 (A)} \cdot u_1(o^*).
    \end{eqnarray*}
    which directly implies \cref{EF reduce result}.
\end{proof}

In the following proofs, with \cref{EF reduce}, we suppose without loss of generality that 
\begin{equation} \label{wlog EF1}
    \forall j \neq i, u_i (A_i) < u_i (A_j).
\end{equation}

With above three assumptions, we turn to characterize the constraints of MNW allocations on the utility function of a specific agent. We use the property that if we move a set $S$ of goods from one agent to another agent under an MNW allocation, the product of their utilities must not increase. Formally, given agent $i, j$, for any $S\subset A_i$ we have
\begin{equation*}
    f_{ij}(S) := u_i(A_i\setminus S)u_j(A_j \cup S) - u_i(A_i)u_j(A_j) \leq 0
\end{equation*}
Further, we have the following observation.
\begin{observation}
    Given agent $i, j$, if there exists $S \subseteq A_i$ such that $f_{ij}(S) > 0$ and $S$ can be partitioned into two non-empty set $S_1$ and $S_2$, then either $f_{ij}(S_1) > 0$ or $f_{ij}(S_2) > 0$.
\end{observation}
This observation directly implies the following corollary.
\begin{corollary}\label{cor:partition covers mnw move condition}
    Given agent $i, j$, for some $S \subseteq M$, if for any $g \in S$ we have $f_{ij}(\{g\}) \leq 0$, then for any subset $T\subseteq S$ we have $f_{ij}(\{g\}) \leq 0$.
\end{corollary}

We extend the above idea to divisible goods to derive a condition on some agent's utility function, the following lemma. 

\begin{lemma} \label{MNW condition on one agent}
    Suppose \cref{no zero utility item}, \cref{wlog no zero agent}, and \cref{wlog EF1} hold. 
    In an MNW allocation, for any two (possibly same) agents $i$ and $j$, 
    \begin{equation} \label{MNW condition on one agent equation}
        \frac{u_i(C_j)}{u_i(A_i)} + \sum_{g \in M_j} \frac{u_i(g)}{u_i(A_i) + u_i(g)} \leq 1.
    \end{equation}
\end{lemma}

\begin{proof}
    If $i = j$, it's sufficient to show that
    \begin{equation*}
        \sum_{g \in M_i} \frac{u_i(g)}{u_i(A_i) + u_i(g)} \leq \sum_{g \in M_i} \frac{u_i(g)}{u_i(A_i)} = \frac{u_i(M_i)}{u_i(A_i)}.
    \end{equation*}
    
    Now assume $i \neq j$. Suppose $P = \{p_1, \cdots, p_k\}$ is an arbitrary partition of $A_j$. The definition of MNW allocation tells that if agent $j$ gives the $p \in P$ to agent $i$, the product of the utilities does not increase. Thus
    \begin{eqnarray}
        && (u_j(A_j) - u_j(p))(u_i(A_i)+u_i(p)) \leq u_j(A_j)u_i(A_i) \nonumber \\
        &\Leftrightarrow&  u_j(A_j)u_i(p) - u_j(p)u_i(A_i) - u_j(p)u_i(p) \leq 0 \nonumber \\
        &\Leftrightarrow&  u_j(A_j)u_i(p)  \leq u_j(p)(u_i(A_i) + u_i(p)) \nonumber \\
        &\Leftrightarrow&  \frac{u_i(p)}{u_i(A_i) + u_i(p)} \leq \frac{u_j(p)}{u_j(A_j)} \label{eq:mnw move condition}.
    \end{eqnarray}
    Notice that this inequality holds for every $p \in P$, we have
    \begin{equation*}
        \sum_{p \in P} \frac{u_i(p)}{u_i(A_i) + u_i(p)} \leq \sum_{p \in P} \frac{u_j(p)}{u_j(A_j)} = 1.
    \end{equation*}
    
    We can improve the above condition by constructing the partition $P$. 
    From \cref{cor:partition covers mnw move condition}, if each of $p_i \in P$ is small enough, the set of inequalities (\ref{eq:mnw move condition}) actually covers almost all the MNW condition on moving a subset $p \subseteq A_j$.
    Specifically, given $\varepsilon>0$, we consider the following partition
    \begin{equation*}
        P = \{ \{g\}: g \in M_j \} \cup P_C,
    \end{equation*}
    where $P_C$ is a partition of $C_j$ such that $|P_C| = \lceil \frac{u_i(C_j)}{\varepsilon} \rceil$ and for every $p \in P_C$, we have $u_i(p) = \frac{u_i(C_j)}{\lceil \frac{u_i(C_j)}{\varepsilon} \rceil} \leq \varepsilon$. 
    Such partition can be obtained through $|P_C|$ queries in the RW model. 
    With this partition, we have
    \begin{eqnarray*}
        &&\sum_{p \in P} \frac{u_i(p)}{u_i(A_i) + u_i(p)} \\
        &=&\sum_{g \in M_j} \frac{u_i(g)}{u_i(A_i) + u_i(g)} + \sum_{p \in P_C} \frac{u_i(p)}{u_i(A_i) + u_i(p)} \\ 
        &=& \sum_{g \in M_j} \frac{u_i(g)}{u_i(A_i) + u_i(g)} + \frac{u_i(p)|P_C| }{u_i(A_i) + u_i(p)} \\
        &=& \sum_{g \in M_j} \frac{u_i(g)}{u_i(A_i) + u_i(g)} + \frac{u_i(C_j)}{u_i(A_i) + u_i(p)} 
        \leq 1.
    \end{eqnarray*}
    Notice from the Sandwich Theorem that
    \begin{equation*}
        \lim_{\varepsilon \rightarrow 0} \frac{u_i(C_j)}{u_i(A_i) + u_i(p)} = \frac{u_i(C_j)}{u_i(A_i)}.
    \end{equation*}
    Hence, as $\varepsilon \rightarrow 0$, we have
    \begin{equation*}
        \sum_{g \in M_j} \frac{u_i(g)}{u_i(A_i) + u_i(g)} + \frac{u_i(C_j)}{u_i(A_i)} \leq 1,
    \end{equation*}
    which completes the proof.
\end{proof}

Given the above nice properties of MNW allocations, it is not hard to prove that these allocations must be PROP$\alpha$.

\begin{proof}[Proof of \cref{MNW is propa}]
    We suppose \cref{no zero utility item}, \cref{wlog no zero agent}, and \cref{wlog EF1} hold.
    Let $A$ be a MNW allocation and denote the best good to be picked from others' bundles by agent $i$ by 
    \begin{equation}
        w_i = \mathop{\max}\limits_{g\in M \backslash M_i} u_i(g).
    \end{equation}
    To finish the proof, it suffices to show that given an agent $i$,
    \begin{equation*}
         u_i(A_i) + \frac{u_i (M)}{u_i (M) + u_i(C)} \cdot w_i \geq \frac{1}{n}.
    \end{equation*}

    Lemma~\ref{MNW condition on one agent} shows that
    \begin{equation*}
        1 \geq \frac{u_i (C_j)}{u_i (A_i)} + \sum_{g \in M_j} \frac{u_i(g)}{u_i(A_i) + u_i(g)}.
    \end{equation*}
    For every $j \neq i$, since $g \in M_j$, we know $u_i(g) \leq w_i$. Thus
    \begin{eqnarray*} 
        1 &\geq& \frac{u_i (C_j)}{u_i (A_i)} + \sum_{g \in M_j} \frac{u_i(g)}{u_i(A_i) + u_i(g)} \\
        &\geq& \frac{u_i (C_j)}{u_i (A_i)} + \sum_{g \in M_j} \frac{u_i(g)}{u_i(A_i) + w_i} \\
        &=& \frac{u_i (C_j)}{u_i (A_i)} + \frac{u_i(M_j)}{u_i(A_i) + w_i}.
    \end{eqnarray*}
    On the other side, if $j = i$, the above inequality trivially holds since
    \begin{equation*}
        1 = \frac{u_i(C_i)}{u_i(A_i)} + \frac{u_i(M_i)}{u_i(A_i)} \geq \frac{u_i (C_i)}{u_i (A_i)} + \frac{u_i(M_i)}{u_i(A_i) + w_i}.
    \end{equation*}
    We then have an upper bound for $u_i(C)$:
    \begin{align} 
        \frac{u_i (C)}{u_i (A_i)} & = \sum_{j \in [n]} \frac{u_i (C_j)}{u_i (A_i)}
        \leq \sum_{j \in [n]} (1 - \frac{u_i(M_j)}{u_i(A_i) + w_i}) \nonumber \\
        &\leq n - \frac{u_i(M)}{u_i(A_i) + w_i}\label{C upper bound},
    \end{align}
    and for $u_i(M)$:
    \begin{equation} \label{M upper bound}
        \frac{u_i(M)}{u_i(A_i) + w_i} = \sum_{j \in [n]} \frac{u_i(M_j)}{u_i(A_i) + w_i} \leq \sum_{j \in [n]} 1 = n.
    \end{equation}
    Therefore,
    \begin{align*}
        &~u_i(A_i) + \frac{u_i (M)}{u_i (M) + u_i(C)} \cdot w_i \\ 
        \geq&~u_i(A_i) + \frac{u_i (M)}{u_i (M) + nu_i(A_i) - \frac{u_i(M)u_i(A_i)}{u_i(A_i) + w_i}} \cdot w_i \tag{by \cref{C upper bound}} \\
        =&~u_i(A_i) + \frac{u_i (M)}{nu_i(A_i) + \frac{w_i u_i (M)}{u_i(A_i) + w_i}} \cdot w_i  \\
        \geq&~u_i(A_i) + \frac{u_i (M)}{nu_i(A_i) + nw_i} \cdot w_i \tag{by \cref{M upper bound}} \\
        =&~u_i(A_i) + \frac{1}{n} \cdot \frac{w_iu_i (M)}{u_i(A_i) + w_i} \\
        =&~\frac{1}{n} \cdot (nu_i(A_i) + u_i(M) - \frac{u_i(A_i)u_i (M)}{u_i(A_i) + w_i})  \\
        \geq&~\frac{1}{n} \cdot (u_i(M) + u_i(C)), \tag{by \cref{C upper bound}}
    \end{align*}
    which completes the proof.
\end{proof}

Finally, we remark that our analysis is tight.

\begin{theorem} \label{thm:MNW only implies propa}
    For any $\varepsilon > 0$, an MNW allocation may not be a PROP$(1-\varepsilon)\alpha$ allocation.
\end{theorem}

\begin{proof}
    Consider the following instance. The allocation $A_1 = C_1, A_2 = M_2, A_3 = M_3$ is one of the MNW allocations. Here, we do not normalize the total utility for better illustration.
    \begin{center}
    \begin{tabular}{cccc}
        \toprule
        & $C_1$ & $|M_2| = \frac{1 + x}{x}$ & $|M_3| = \frac{1 + x}{x}$ \\
        \midrule
        $u_1(\cdot)$ & 1 & $x, \forall o \in M_2$ & $x, \forall o \in M_3$  \\
        $u_2(\cdot)$ & 0 & $x, \forall o \in M_2$ & $x, \forall o \in M_3$ \\
        $u_3(\cdot)$ & 0 & $x, \forall o \in M_2$ & $x, \forall o \in M_3$ \\
        \bottomrule
    \end{tabular}
    \end{center}
    For agent 1, the PROP$f(\alpha)$ requires
    \begin{equation*}
        1 + f(\frac{2+2x}{3+2x})\cdot x \geq \frac{3+2x}{3},
    \end{equation*}
    which is equivalent to
    \begin{equation*}
        f(\frac{2+2x}{3+2x}) \geq \frac{2}{3}.
    \end{equation*}
    Note that $\lim_{x \rightarrow 0} \frac{2+2x}{3+2x} = \frac{2}{3}$. For any $\varepsilon > 0$, when $f(x) = x - \varepsilon$ or $f(x) = (1-\varepsilon)x$, MNW both failed to satisfy PROP$f(\alpha)$ as $x \rightarrow 0$ in this example.

    One may argue that there are other MNW allocations in the example above. We could check the following instance. The allocation $A_1 = C_1, A_2 = M_2, A_3 = M_3$ is the only MNW allocation. 
    \begin{center}
    \begin{tabular}{cccc}
        \toprule
        & $C_1$ & $|M_2| = \frac{1 + x}{x}$ & $|M_3| = \frac{1 + x}{x}$ \\
        \midrule
        $u_1(\cdot)$ & 1 & $x - x^3, \forall o\in M_2$ & $x - x^3, \forall o\in M_3$ \\
        $u_2(\cdot)$ & 0 & $x, \forall o\in M_2$ & $x, \forall o\in M_3$ \\
        $u_3(\cdot)$ & 0 & $x, \forall o\in M_2$ & $x, \forall o\in M_3$ \\
        \bottomrule
    \end{tabular}
    \end{center}
    Now we have
    \begin{equation*}
        1 + f(\frac{2+2x-2x^2-2x^3}{3+2x-2x^2-2x^3})\cdot x \geq \frac{3+2x-2x^2-2x^3}{3},
    \end{equation*}
    or
    \begin{equation*}
        f(\frac{2+2x-2x^2-2x^3}{3+2x-2x^2-2x^3}) \geq \frac{2-2x-2x^2}{3}.
    \end{equation*}
    When $x \rightarrow 0$ we have
    \begin{equation*}
        f(\frac{2}{3}) \geq \frac{2}{3},
    \end{equation*}
    which fails to satisfy \propfa~when $f(x) = x - \varepsilon$ or $f(x) = (1-\varepsilon)x$.
\end{proof}

\section{Relation with Other Fairness Notions}\label{sec:extension}

We explore the relation among \effa, \propfa~and other notions (EFM and MMS) in the mixed goods setting. A summary of the results in this section is provided in Figure~\ref{fig:relation}.

\begin{figure}[ht]
    \centering    \includegraphics[width=0.45\textwidth]{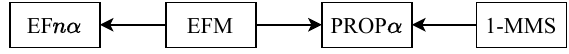} 
    \caption{Relations among different fairness notions, where $X \rightarrow Y$ means that an allocation satisfying $X$ must also satisfy $Y$. }
    \label{fig:relation}
\end{figure}

\subsection{Connections to EFM}
We first discuss the relations of our ``up to a fraction'' fairness notions with EFM, proposed in \citep{BeiLiLi20}. 

\begin{definition}[EFM]\label{def:EFM}
    An allocation $\mathcal{A}$ is said to be \emph{envy-free for mixed goods (EFM)} if for any two agents $i, j \in N$:
    \begin{itemize}
        \item if $C_j = \emptyset$ and $M_j \neq \emptyset$, there exists $o \in M_j$ such that $u_i(A_i) \geq u_i(A_j \backslash \{o\})$,
        \item otherwise, $u_i(A_i) \geq u_i(A_j)$.
    \end{itemize}
\end{definition}

It is easy to verify that EFM does not imply EF$\alpha$; recall the example in Section \ref{sec:intro}.
However, EFM can imply a generalized version of EF$\alpha$, as shown in the theorem below.

\begin{theorem} \label{thm:efm implies efna}
    Any EFM allocation is EF$n\alpha$.
\end{theorem}

\begin{proof}
    Fix an agent $1$ arbitrarily. If $\alpha_1 \ge \frac{1}{n}$, EF$n\alpha$ holds since EF1 is guaranteed by EFM. Otherwise, we have $u_1(C) > \frac{n-1}{n}$ and $u_1(M) < \frac{1}{n}$. 
    
    If $u_1(A_1) \geq \frac{1}{n}$, suppose $u_1(A_1) < u_1(A_i)$, EFM tells $C_i = \emptyset$ and thus $u_1(M_i) = u_1(A_i) > u_1(A_1) \geq \frac{1}{n} $, which contradicts with $u_1(M) < \frac{1}{n}$. Therefore, EF holds in this case.

    We then consider the case when $u_1(A_1) < \frac{1}{n}$.
    Since $u_1(C)>\frac{n-1}{n}$ and any agent $i$ containing some pieces of cake must have $u_1(C_i)\le u_1(A_i)<u_1(A_1)<\frac1n$ by the definition of EFM, $u_1(C_i) > 0$ must hold for all $i\in N$ otherwise there must exist some divisible goods which cannot be allocated.
    Thus, agent 1 must not envy any other agent.
\end{proof}

On the other side, we show that EF$n\alpha$ is the best guarantee under EF$f(\alpha)$ that an EFM allocation ensures.

\begin{theorem} \label{thm:efm only implies efna}
    For any $\varepsilon > 0$, an EFM allocation may not be EF$(n-\varepsilon)\alpha$.
\end{theorem}

\begin{proof}
    Let $x = \frac{\varepsilon}{n}$. Consider the following instance with $n$ agents, $\frac{1}{x}$ indivisible goods and one cake. 
    Here, without loss of generality, we assume $\frac1x$ is an integer, otherwise we can choose another $\varepsilon'=\frac{1}{\lceil \frac{1}{\varepsilon}\rceil}$ and use the corresponding instance under this $\varepsilon'$.
    we do not normalize the total utility for better illustration here either.
    \begin{center}
    \begin{tabular}{cccc}
        \toprule
        & $|M| = \frac{1}{x}$ & $C$ & $\alpha$ \\
        \midrule
        $u_i(\cdot), \forall i \in [n]$ & $x, \forall o \in M$ & $(n-1)(1-x)$ & $\frac{1}{(1-x)n+x}$ \\
        \bottomrule
    \end{tabular}
    \end{center}
    By allocating $M$ to agent $1$ and dividing $C$ equally among the rest of the agents, we have $u_i(A_1) = 1$ and $u_i(A_j) = 1-x$ for any $i\in[n]$ and $j \in [n] \backslash \{1\}$. It's easy to check that this is an EFM allocation. However, for any good $g \in A_1$,
    \begin{eqnarray*}
        && u_2(A_1) - (n-\varepsilon)\alpha_2 \cdot u_2(g) \\
        &=& 1 - (n - \varepsilon) \cdot \frac{1}{(1-x)n+x} \cdot x \\
        &>& 1 - \frac{n - \varepsilon}{(1-x)n} \cdot x = u_2(A_2).
    \end{eqnarray*}
    This implies that the condition of EF$(n-\varepsilon)\alpha$ from agent $2$ to agent $1$ is not satisfied.
\end{proof}

We then consider the relation between PROP$\alpha$ and EFM.

\begin{theorem} \label{thm:efm implies propa}
    An EFM allocation is \propa.
\end{theorem}

\begin{proof}
    Fix an agent $1$ arbitrarily. If $u_1(A_1)\ge \frac{1}{n}$, agent 1 achieves PROP (and thus PROP$\alpha$). Otherwise, we assume $u_1(A_1)<\frac1n$ and pick the good $w\notin M_1$ such that $u_1(w)=\max_{g\notin M_1}u_1(g)$.

    We assume an integer $k$ such that $u_1(A_1)\ge u_1(A_j)$ for all $j\in[k]$ and $u_1(A_1)<u_1(A_j)$ for all $j>k$. By EFM, we have: $u_1(A_1)+u_1(w)\ge u_1(A_j)$ for all $j>k$.
    Adding all these inequalities for all $j\in[n]$, we have:
    $u_1(A_1)+\frac{n-k}{n}u_1(w)\ge \frac1n$.

    If $u_1(M)\ge \frac{n-k}{n}$, the above inequality ensures PROP$\alpha$. 
    If not, we have $u_1(C)>\frac{k}{n}$. Since $\frac1n>u_1(A_1)\ge u_1(A_j)$ for all $j\in[k]$, there exists some cake allocated to agent $j>k$, a contradiction with the definition of EFM.
\end{proof}

This relation is also tight due to the following result.

\begin{theorem} \label{thm:efm only implies propa}
    For any $\varepsilon > 0$, an EFM allocation may not be PROP$(1-\varepsilon)\alpha$.
\end{theorem}

\begin{proof}
    Let $x = (n-1)\varepsilon$. Consider the following instance with $n$ agents, $\frac{n-1}{x}$ indivisible goods, and one cake. Here, without loss of generality, we assume $\frac1x$ is an integer and $x\le 1$, otherwise we can choose another $\varepsilon'=\frac{1}{\lceil \frac{1}{\min\{1/(n-1),\varepsilon\}}\rceil}$ and use the corresponding instance under this $\varepsilon'$.
    we do not normalize the total utility for better illustration here either.
    \begin{center}
    \begin{tabular}{cccc}
        \toprule
        & $|M| = \frac{n-1}{x}$ & $C$ & $\alpha$ \\
        \midrule
        $u_i(\cdot), \forall i \in [n]$ & $x, \forall o \in M$ & $1-x$ & $\frac{n-1}{n-x}$ \\
        \bottomrule
    \end{tabular}
    \end{center}
    By allocating $C$ to agent $1$ and dividing $M$ equally among the rest of the agents, we have $u_i(A_1) = 1-x$ and $u_i(A_j) = 1$ for any $i\in[n]$ and $j \in [n] \backslash \{1\}$. It is easy to check that this is an EFM allocation. However, for any good $g \in A \backslash A_1$,
    \begin{eqnarray*}
        && u_1(A_1) + (1-\varepsilon)\alpha_1 \cdot u_1(g) \\
        &=& 1 - x + (1-\varepsilon)\cdot \frac{n-1}{n-x} \cdot x \\
        &=& 1 - (\frac{1-x}{n-x} + \frac{n-1}{n-x}\varepsilon) \cdot x \\
        &=& 1 - \frac{x}{n-x} \\
        &<& 1 - \frac{x}{n} = \frac{n-x}{n}.
    \end{eqnarray*}
    This implies that the condition of PROP$(1-\varepsilon)\alpha$ for agent $1$ is not satisfied.
\end{proof}

It is worth noting that \citet{BeiLiLi20} designed an algorithm for computing an EFM allocation. However, their algorithm utilizes the {\em perfect allocation oracle}, which is not in polynomial time when we have a heterogeneous cake, and consists of the intricate envy-graph maintenance and envy-cycle elimination subroutine. On the contrary, \cref{alg:propalpha} runs in polynomial time and is simple to implement.

\subsection{Connections to MMS}
We also consider the relation of our ``up to a fraction'' fairness notions with MMS as defined in the following.

\begin{definition}[$\beta$-MMS]\label{def:MMS}
    Let $\Pi_n(A)$ be the set of all $n$-partitions of $A$. 
    The maximin share (MMS) of any agent $i \in N$ is defined as
    \[
    \mathrm{MMS}_i = \max_{\mathcal{P} = (P_1, P_2, \cdots, P_k) \in \Pi_n(A)} \min_{j \in N} u_i(P_j).
    \]   
    An allocation that reaches $\mathrm{MMS}_i$ is called an MMS-allocation of agent $i$. 
    Given any $\beta \in [0,1]$, allocation $\mathcal{A}$ is $\beta$-approximate MMS fair ($\beta$-MMS) if $u_i(A_i) \geq \beta \cdot \mathrm{MMS}_i$ for every agent $i \in N$.
    When $\beta = 1$, we simply
    write MMS.
\end{definition}

It is easy to see that when the goods are all divisible, MMS coincides with PROP. 
When the goods are all indivisible, MMS is strictly weaker than PROP but implies PROP1 \citep{caragiannis2023new}. 
Our next result is a generalization that encompasses these two extreme cases.

\begin{theorem} \label{thm:MMS implies propa}
   Any MMS allocation is \propa.
\end{theorem}

\begin{proof}
For an arbitrary MMS allocation $X=\{X_1,\ldots,X_n\}$, we assume this allocation $X$ does not satisfy \propa, i.e., there exists some $i\in[n]$ s.t. $u_i(X_i)<\frac1n-\alpha_iu_i(o)$ for every indivisible good $o\notin X_i$.
Let $Y=\{Y_1,\ldots,Y_n\}$ as the MMS-allocation of agent $i$ which minimizes the number of the bundles with the utility $\mathrm{MMS}_i$ from $i$'s perspective. Without loss of generality, we assume $Y_1$ is the bundle with the largest fraction of the cake $C$ among all bundles with the utility $\mathrm{MMS}_i$ and assume $Y_n$ is the bundle with the largest utility under $u_i$.

Since $X$ is an MMS allocation, $\frac1n>u_i(X_i)\ge \mathrm{MMS}_i=u_i(Y_1)$, we then have $u_i(Y_n)>\frac1n$.
If there exists a bundle $Y_i$ with $u_i(Y_i)>\mathrm{MMS}_i$ which contains some pieces of cake, we can move some of them to the bundle $Y_1$ which can make the allocation $Y$ is no longer the required MMS-allocation as said above.
Thus, the whole cake is shared by the bundles with utility $\mathrm{MMS}_i$, which means that $u_i(Y_1)\ge u_i(C_1)\ge \frac{u_i(C)}{n}$ and $Y_n$ contains only indivisible goods. Here, $C_1$ represents the cake in $Y_1$.

Because $u_i(Y_n)>\frac1n>u_i(X_i)$ and $Y_n$ consists of no cake, there exists a good $o\in Y_n$ such that $o\notin X_i$.
If $u_i(o)\ge \frac1n$, we have $u_i(X_i)+\alpha_iu_i(o)\ge u_i(Y_1)+\frac{u_i(M)}{n}\ge \frac{u_i(C)}{n}+\frac{u_i(M)}{n}=\frac1n$, which leads to a contradiction.
If $u_i(o)\le u_i(C_1)$ (so it is less than $u_i(Y_n)$), we can exchange the good $o$ and an equivalent fraction of $C_1$ and then $Y_n$ becomes a bundle with $u_i(Y_n)>\mathrm{MMS}_i$ which contains some pieces of cake, which also leads to a contradiction with the definition of the allocation $Y$.

For the last case where $\frac1n>u_i(o)>u_i(C_1)$, we have $u_i(Y_1)+u_i(o)-u_i(C_1)\le u_i(X_i)+u_i(o)-\frac{u_i(C)}{n}<u_i(X_i)+u_i(o)-u_i(o)u_i(C)=u_i(X_i)+u_i(o)u_i(M)<\frac1n$.
The last inequality is from the fact that $X_i$ violates the \propa ~condition.
With the fact that $u_i(Y_n)>\frac 1n$, we can exchange the good $o$ and $C_1$ to reach an MMS allocation with a smaller number of the bundles with utility $\mathrm{MMS}_i$, which violates the definition of $Y$.
This completes our proof.
\end{proof}

Recall that MMS allocations may not exist, however, an approximately MMS allocation may not be \propa.

\begin{theorem} \label{thm:approx MMS is not propa}
For any $\beta \in (0,1)$, a $\beta$-MMS allocation may not be \propa.
\end{theorem}

\begin{proof}
    One example is to allocate the cake to $n$ agents. Now an MMS allocation is also PROP (and PROP$\alpha$ since $\alpha_i = 0$ for all $i \in [n]$). Any approximately MMS allocation fails to satisfy MMS and thus PROP$\alpha$.

    Another example is allocating many small goods to $n$ agents. For instance, $\frac{n}{x}$ goods with utility $x$ for all agents. Here for any agent $i$, $\mathrm{MMS}_i = 1$ and $\alpha_i = 1$. By giving the first agent $\frac{1}{x}-2$ goods and allocating the rest of the goods evenly to the rest of the agents (so that each agent obtains at least $\frac{1}{x}$ goods), this allocation is $(1-2x)$-MMS. However, it's obvious that this is not a PROP1 allocation.
\end{proof}

\subsection{Connections to PROPm}
We discuss the relation of our ``up to a fraction'' fairness notion with a variant of \emph{proportionality up to the maximin item} (PROPm) in~\citep{baklanov2021propm}. PROPm is a notion more demanding than PROP1 in the indivisible goods setting as defined below.


\begin{definition}[PROPm~\citep{baklanov2021propm}]
    An allocation is said to be \emph{proportionality up to the maximin item} (PROPm) if $u_i(A_i)+ \max_{j\neq i} \min_{g\in A_i} {u_i(g)} \geq 1/n$.
\end{definition}

In the mixed goods setting, PROP$\alpha$ is parallel to PROPm. Briefly, PROPm strengthens PROP1 by limiting the set of goods that can be chosen to achieve PROP in the indivisible goods setting. Comparatively, our ``up to a fraction'' fairness notions of PROP$\alpha$ strengthens PROP1 by demanding less fraction of such a chosen good in the mixed goods setting. For the mixed goods setting, one possible way to extend PROPm is defined below.

\begin{definition}[PROPmM]
    An allocation is said to be \emph{proportionality up to the maximin item for mixed goods} \emph{PROPmM} in the mixed goods setting if $v_i(A_i)+ \max_{j\neq i,C_j=\emptyset} \min_{g\in M_i} {v_i(g)} \geq 1/n$.
\end{definition}

Specifically, PROPmM allows an agent to choose a good that is the worst in another agent's bundle. By treating the cake as multiple infinitesimal indivisible goods, if some cake is included in a bundle, such worst good would be infinitesimal and yield a value of $0$. Thus, we allow each agent to only choose the smallest items that are not in a bundle containing cake. Notice that this definition is stronger than just substituting the $\min_{g\in A_i}$ with $\min_{g\in M_i}$.

We then compare PROP$\alpha$ and PROPmM.
First, PROP$\alpha$ may be fairer than PROPmM when indivisible items are of similar value. Consider the following instance with 2 identical agents, 3 indivisible goods, and one cake. Allocate the cake to agent 1 and all the indivisible goods to agent 2. This allocation is PROPmM but not PROP$\alpha$.

\begin{table}[h]
    \centering
    \begin{tabular}{c|cc|c}
        \toprule
        &  $g \in M, |M|=3$ & $C$ & $\alpha$ \\
        \midrule
        $u(\cdot)$ &  $0.25$ & $0.25$ & $0.75$\\
        \bottomrule
    \end{tabular}
\end{table}

On the other side, when the cake is small, PROP$\alpha$ is almost PROP1, and PROPmM can be fairer. 
To see an example, consider the following instance with 2 identical agents, 2 indivisible goods, and one cake. Allocate $g_1$ to agent 1 and the rest to agent 2. This allocation is PROP$\alpha$ but not PROPmM.

\begin{table}[h]
    \centering
    \begin{tabular}{c|ccc|c}
        \toprule
        &  $g_1$ & $g_2$ & $C$ & $\alpha$ \\
        \midrule
        $u(\cdot)$ &  $0.4$ & $0.4$ & $0.2$ & $0.8$\\
        \bottomrule
    \end{tabular}
\end{table} 

Therefore, we can conclude that PROP$\alpha$ is incomparable with PROPmM.
We leave the comparisons with other notable notions, e.g., \emph{envy-freeness up to one less-preferred good} (EFL) in~\citep{barman2018groupwise}, as a promising future work.

\section{Conclusion} \label{sec:conclusion}
We study the fair allocation of a mixture of divisible and indivisible goods.
We introduce the indivisibility ratio and fairness notions of envy-free and proportional up to a fractional good, which serves as a smooth connection between EF/PROP and EF1/PROP1.
Our results exhibit the limit of the amount of the fractional item that we need to relax so that a fair allocation is guaranteed, which affirm
the intuition that the more divisible items we have, the
fairer allocations we can achieve. 
There are some problems left open.
For example, there is a constant gap between the upper and lower bounds of the fractional relaxation of EF, and it is not clear whether EF$n\alpha$ and PO are compatible. 
Our paper also unveils intriguing possibilities for future research. One such avenue is proposing alternative relaxations of the ideal fairness principles to better capture the characteristics of mixed scenarios, such as the customized indivisibility ratio in our model. 

\section*{Acknowledgments}
Shengxin Liu is funded by the National Natural Science Foundation of China (No. 62102117), by the Shenzhen Science and Technology Program (No. GXWD20231129111306002), and by the Guangdong Basic and Applied Basic Research Foundation (No. 2023A1515011188).
Bo Li is funded by the National Natural Science Foundation of China (No. 62102333) and Hong Kong SAR Research Grants Council (No. PolyU 25211321).

\bibliographystyle{named}
\bibliography{ijcai24}

\end{document}